\documentclass[a4paper,UKenglish,cleveref, autoref]{lipics-v2019}
\usepackage[utf8]{inputenc}
\usepackage{amsmath}
\usepackage{amsthm,comment}
\usepackage{tikz}
\usepackage{slashbox}
\usepackage{boxedminipage}
\usepackage{mathtools}
\usepackage{color}

\newcommand{\NP}{{\sf NP}}

\title{Hard Problems That Quickly Become Very Easy} 

\titlerunning{On the Complexity of Graph Problems for Hereditary Graph Classes}

\author{Barnaby Martin}{Department of Computer Science, Durham University, United Kingdom}{barnaby.d.martin@durham.ac.uk}{}{}
\author{Dani\"el Paulusma}{Department of Computer Science, Durham University, United Kingdom}{daniel.paulusma@durham.ac.uk}{0000-0001-5945-9287}{supported by the Leverhulme Trust (RPG-2016-258).}
\author{Siani Smith}{Department of Computer Science, Durham University, United Kingdom}{siani.smith@durham.ac.uk}{}{}

\authorrunning{B. Martin, D. Paulusma and S. Smith }

\Copyright{Barnaby Martin, Daniel Paulusma and Siani Smith}



\nolinenumbers 
\hideLIPIcs

\newcommand{\problemdef}[3]{
	\begin{center}
		\begin{boxedminipage}{.99\textwidth}
			\textsc{{#1}}\\[2pt]
			\begin{tabular}{ r p{0.8\textwidth}}
				\textit{~~~~Instance:} & {#2}\\
				\textit{Question:} & {#3}
			\end{tabular}
		\end{boxedminipage}
	\end{center}
}

\begin{document}

\maketitle

\begin{abstract}
A graph class is hereditary if it is closed under vertex deletion.
We give examples of \NP-hard, PSPACE-complete and NEXPTIME-complete
problems that become constant-time solvable for every hereditary graph class that is not equal to the class of all graphs.
\end{abstract}


\section{Introduction}\label{s-intro}

Many discrete optimization problems that can be defined on graphs are computationally hard but may become tractable if we restrict the input to some special graph class. One of the most natural properties of a graph class is to be closed under vertex deletion. 
A graph class with this property is also called {\it hereditary}. Hereditary graph classes contain many well-studied graph classes, and moreover, they enable a systematic study of a graph problem under input restrictions, as we explain below. 

For some graph $H$, a graph $G$ is {\it $H$-free} if $G$ does not contain $H$ as an induced subgraph, that is, $G$ cannot be modified into $H$ by a sequence of vertex deletions. For a set of graphs ${\cal H}$, a graph $G$ is
{\it ${\cal H}$-free} if $G$ is $H$-free for every $H\in {\cal H}$. It is well known (and not difficult to see) that a graph class ${\cal G}$ is hereditary if and only if ${\cal G}$ is ${\cal F}_{\cal G}$-free for some unique set  ${\cal F}_{\cal G}$ of minimal forbidden induced subgraphs.\footnote{The set ${\cal F}_{\cal G}$ may have infinite size. For example, if ${\cal G}$ is the class of bipartite graphs, then ${\cal F}_{\cal G}$ consists of all cycles of odd length.}
For a systematic study, one may first focus on hereditary graph classes ${\cal G}$ for which ${\cal F}_{\cal G}$ has small size.  Apart from developing new methodology and obtaining a deeper understanding of computational hardness, this should ideally lead to {\it dichotomy theorems}. Such theorems tell us exactly under which input restrictions a graph problem can be solved efficiently and under which restrictions the problem stays computationally hard.

The borderline between hardness and tractability is often far from clear beforehand and jumps in computational complexity can be extreme. In order to illustrate this kind of behaviour we gave, in~\cite{MPS19}, an example of an {\it \NP-hard graph} problem that becomes {\it constant-time} solvable for {\it every} hereditary graph class that is not equal to the class of all graphs. This problem was related to vertex colouring and being a universal graph (a graph is {\it universal} for some graph class~${\cal G}$ if it contains every graph from ${\cal G}$ as an induced subgraph).

In this note we expand on the extremities in computational complexity jumps.
In Section~\ref{s-ex1} we recall the above \NP-hard example of~\cite{MPS19}, as this example shows the essence of our constructions (we now also show that this problem belongs to $\Sigma^{\mathrm{P}}_2$). Then, in Section~\ref{s-ex2}, we give another example of such a problem, which is related to vertex colouring reconfiguration and universal graphs. However, for this problem we can prove that it is even PSPACE-complete for general graphs. Finally, in Section~\ref{s-ex3}, we move away from problems defined on graphs and introduce a circuits problem that is even NEXPTIME-complete in general; note that for the latter problem no polynomial-time algorithm can exist, whereas for the first two problems such a claim only holds subject to standard complexity assumptions.

Throughout the paper, the notion of graph colouring plays an important role. A {\it $k$-colouring} of $G$ is a function $c:V\to \{1,\ldots,k\}$ such that for every two adjacent vertices $u$ and $v$ we have that $c(u)\neq c(v)$. We also say that $c(u)$ is the {\it colour} of $u$. For a fixed $k$ (that is, $k$ is not part of the input), the corresponding decision problem {\sc $k$-Colouring} is to decide if a given graph has a $k$-colouring. See~\cite{GJPS17} for a survey on graph colouring for hereditary graph classes.
The {\it disjoint union} $G+H$ of two graphs $G$ and $H$ is the graph $(V(G)\cup V(H),E(G)\cup E(H))$.

\section{NP-Hardness: Graph Colouring}\label{s-ex1}

We define the following decision problem and then present our first result, whose proof will serve as a basis for our other proofs.

\problemdef{{\sc Colouring-or-Subgraph}}{an $n$-vertex graph $G$}{is $G$ $\lceil \sqrt{\log n} \rceil$-colourable or $H$-free for some graph $H$ with $|V(H)|\leq  \lceil \sqrt{\log n} \rceil$?} 

\begin{figure} [h]
	\resizebox{8.5cm}{!} {
			\begin{tikzpicture}[main_node/.style={circle,draw,minimum size=1cm,inner sep=3pt]}]
\node[main_node](g1) at (0,0){};
\node[main_node](g2) at (1.5,0){};
\node[main_node](g3) at (3,0){};
\draw (-1,-1) rectangle (4,2);
\node[main_node](p1) at (0,3){};
\node[main_node](p2) at (1.5,3){};
\node[main_node](p3) at (3,3){};
\node(g) at (-2,0) {\Huge$G$};
\node[circle, draw, minimum size=2.5cm](H1) at (7,0){\Huge$H_1$};
\node[circle, draw, minimum size=2.5cm](H2) at (10,0){\Huge$H_2$};
\node[circle,draw, minimum size=2.5cm](H3) at (13,0){$\dots$};
\node[circle, draw, minimum size=2.5cm](H4) at (16,0){\Huge$H_r$};
\draw(p1)--(g1);
\draw(p1)--(g2);
\draw(p1)--(g3);
\draw(p2)--(g1);
\draw(p2)--(g2);
\draw(p2)--(g3);
\draw(p3)--(g1);
\draw(p3)--(g2);
\draw(p3)--(g3);
\draw(p1)--(p2)--(p3);
\draw(p1) to [out=90, in=90] (p3);

\end{tikzpicture}
}
\caption{An example of a graph $G^{\prime}$,where $p=6$. The graph $G^*$ is the connected component on the left. Not all vertices of the subgraph $G$ of $G^*$ are drawn and not all graphs $H_i$ on $p$ vertices are displayed.}\label{f-ex}
\end{figure}
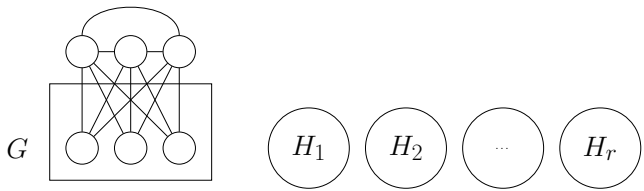

\begin{theorem}\label{t-0}
The {\sc Colouring-or-Subgraph} problem is \NP-hard, but constant-time solvable for every hereditary graph class not equal to the class of all graphs.
\end{theorem}

\begin{proof}
To prove \NP-hardness we reduce from {\sc $3$-Colouring}, which we recall is \NP-complete~\cite{Lo73}. 
Let $G$ be an $n$-vertex graph. Set 
$p=\lceil \sqrt{\log 3n} \rceil$. 
We may assume without loss of generality that $p\geq 4$.
Add $p-3$ pairwise adjacent vertices to $G$. Make the new vertices also adjacent to every vertex of $G$. 
We denote the new graph by $G^*$.
Let $\{H_1,\ldots,H_r\}$ be the set of all graphs with exactly $p$ vertices.
We now define the graph $G'$ as the disjoint union of $G^*$ and the graphs $H_1,\ldots,H_r$; 
see also Figure~\ref{f-ex}.
Note that the number of vertices of the graph $H_1+\ldots+H_r$ is at most
$p2^{\frac{p(p-1)}{2}} \leq \lceil\sqrt{\log 3n} \rceil \cdot \sqrt{3n}\leq n$ as $p\geq 4$.  
This implies that the
number of vertices in $G'$ is
$$|V(G')|=|V(G)|+p-3+|V(H_1)|+\ldots + |V(H_r)| \leq n+(p-3)+ n < 3n.$$
\noindent
In particular, the above shows that the number of vertices in $G'$ is bounded by a polynomial in $n$.
We now add $3n-|V(G'|$ isolated vertices to $G'$ such that $G'$ has exactly $3n$ vertices.

We claim that $G$ is $3$-colourable if and only if $G'$ is a is a yes-instance of {\sc Colouring-or-Subgraph}.
First suppose that $G$ is $3$-colourable. We give each of the $p-3$ vertices of $G^*$ that is not in $G$ a unique colour from $\{4,\ldots,p\}$. As $G$ is $3$-colourable, we find that 
$G^*$ is $p$-colourable. 
As $G'$ is the disjoint union of $G^*$ and the graphs $H_1,\ldots,H_r$ 
(and some isolated vertices, which we give colour~$1$),
we must now consider the graphs $H_1,\ldots,H_r$.
By construction, each $H_i$ has $p$ vertices, so we can give each vertex of each $H_i$ a colour from $\{1,\ldots,p\}$ that is not used on any other vertex of $H_i$. Hence, we find that $G'$ is $p$-colourable. 
As 
$p=\lceil \sqrt{\log 3n} \rceil=\lceil \sqrt{\log |V(G')|} \rceil$, 
this implies that $G'$ is a yes-instance of {\sc Colouring-or-Subgraph}.

Now suppose that $G'$ is a yes-instance of {\sc Colouring-or-Subgraph}. 
Recall that $G'$ is the disjoint union of the graph $G^*$, the graphs $H_1,\ldots,H_r$ and some isolated vertices, and recall also that the graphs
$H_1,\ldots,H_r$ are all the graphs on exactly $p$ vertices.
Hence, $G'$ contains every graph on at most $p$ vertices as an induced subgraph.
In other words, $G'$ is not $H$-free for some graph $H$ with $|V(H)|\leq p$.
As $G'$ is a yes-instance of {\sc Colouring-or-Subgraph} and $p=\lceil \sqrt{\log 3n} \rceil=\lceil \sqrt{\log |V(G')|} \rceil$, this means that $G'$ must be $p$-colourable. 
As the $p-3$ vertices of $V(G^*)\setminus V(G)$ form a clique,
we may assume without loss of generality that they are coloured $4,\ldots,p$, respectively. All these $p-3$ vertices are adjacent to every vertex of $G$ in $G^*$. Consequently, every vertex of $V(G)$ 
must have received a colour from the set $\{1,2,3\}$. Hence, $G$ is $3$-colourable.
 
\medskip
\noindent
We now prove the second part of the theorem. 
Let ${\cal G}$ be a hereditary graph class that is not the class of all graphs.
Then there exists at least one graph $H$ such that every graph $G\in {\cal G}$ is $H$-free. 
Let $\ell=|V(H)|$.
We claim that {\sc Colouring-or-Subgraph} is constant-time solvable for ${\cal G}$. Let $G\in {\cal G}$ be an $n$-vertex graph.
If $n \leq 2^{\ell^2}$, then $G$ has constant size and the problem is constant-time solvable.
If $n>2^{\ell^2}$, then
$$|V(H)|=\ell< \sqrt{\log n}  \leq \lceil \sqrt{\log n} \rceil.$$
Hence $G$ is a yes-instance of {\sc Colouring-or-Subgraph}, as $G$ is $H$-free and $H$ has at most $\lceil \sqrt{\log n} \rceil$ vertices.
\end{proof}

\noindent
We do not know if {\sc Colouring-or-Subgraph} is in $\NP=\Sigma^{\mathrm{P}}_1$. The problem arises when we try to check if an input $G$ is $H$-free for some particular $H$, of size (say) $\lceil \sqrt{\log n} \rceil$, which takes time $n^{\lceil \sqrt{\log n} \rceil}$ by brute force. Note that it is crucial for our proof that the size of $H$ depends on a function of $n$. 
We can however show that the problem belongs to the class $\Sigma^{\mathrm{P}}_2$ 
(a language $L$ is in $\Sigma^{\mathrm{P}}_2$ if there exists a polynomial-time predicate $P$ and a polynomial $q$ such that a string $x$ belongs to $L$ if and only if there exists a string $y$ of length $q(|x|)$ such that for every string $z$ of length $q(|x|)$, $P(x,y,z)=1$).

\begin{theorem}
{\sc Colouring-or-Subgraph} is in $\Sigma^{\mathrm{P}}_2$.
\end{theorem}

\begin{proof}
We can verify whether an input $G$ is $\lceil \sqrt{\log n} \rceil$-colourable in NP. Let us explain how to verify if $G$ is $H$-free for some graph $H$ with $|V(H)|\leq  \lceil \sqrt{\log n} \rceil$. Plainly, we can guess existentially the graph $H$ whose vertices are ordered $u_1,\ldots,u_{|V(H)|}$. Now we guess universally $|V(H)|$ vertices $v_1,\ldots,v_{|V(H)|}$ in $G$. Finally, we test whether the respective map of $u_1,\ldots,u_{|V(H)|}$ to $v_1,\ldots,v_{|V(H)|}$ has the property that $u_iu_j$ is an edge in $H$ if and only if $v_iv_j$ is an edge in $G$. The latter can be accomplished in polynomial time and we are done.
\end{proof}

\section{PSPACE: Graph Colouring Reconfiguration}\label{s-ex2}

Let $G=(V,E)$ be a graph. A {\it clique} is a set of pairwise adjacent vertices in $G$. The set of neighbours of a vertex $v\in V$ is denoted by $N_G(v)=\{u\; |\; uv\in E\}$.
The {\it $k$-colouring reconfiguration graph}~$R_k(G)$ of $G$ is the graph whose vertices are $k$-colourings of $G$ and two vertices are adjacent if and only if the two corresponding $k$-colourings differ on exactly one vertex of $G$.  
In the following problem, $k$ is a fixed constant, that is, $k$ is not part of the input.

\problemdef{$k$-Colour-Path}{A graph $G$ with two $k$-colourings $\alpha$ and $\beta$.}{Does $R_k(G)$ contain a path from $\alpha$ to $\beta$?}

Bonsma and Cereceda~\cite{BC09} proved that {\sc $3$-Colour-Path} is polynomial-time solvable, but for $k\geq 4$ they showed the following result, which holds even for bipartite graphs and which we will need in the next section.

\begin{theorem}[\cite{BC09}]\label{t-bc}
For every integer $k\geq 4$, the {\sc $k$-Colour-Path} problem is \emph{PSPACE}-complete.
\end{theorem}

We define the following problem.

\problemdef{{\sc Colour-Path-or-Subgraph}}{an $n$-vertex graph $G$ with two $p$-colourings $\alpha$ and $\beta$ for $p=\lceil \sqrt{\log n} \rceil$.}{Does $R_p(G)$ contain a path from $\alpha$ to $\beta$, or does there exist a graph $H$ with $|V(H)|\leq p$ such that $G$ is $H$-free?}  
 
\begin{theorem}\label{t-main1}
{\sc Colour-Path-or-Subgraph} is \emph{PSPACE}-complete, but constant-time solvable for every hereditary graph class not equal to the class of all graphs.
\end{theorem}

\begin{proof}
Let $(G,\alpha,\beta)$ be an instance of {\sc Colour-Path-or-Subgraph}, 
where $G$ is a graph on $n$ vertices. 
Set $p=\lceil \sqrt{\log n} \rceil$.
We can check if $R_p(G)$ contains a path from $\alpha$ to $\beta$ using a polynomial amount of space using the same proof as used in Theorem~\ref{t-bc} for {\sc $4$-Colour-Path}~\cite{BC09}. So we first prove membership to NPSPACE.
As a certificate we can take a sequence of $p$-colourings of $G$ and check in polynomial space if this sequence 
is an $\alpha-\beta$ path in $R_p(G)$: check if the first $p$-colouring is $\alpha$; then check if the next $p$-colouring differs exactly at one place from the previous $p$-colouring and if so delete the previous $p$-colouring and continue; finally check if the last $p$-colouring is $\beta$. Now, as 
PSPACE$=$NPSPACE due to Savitch's Theorem~\cite{Sa70}, this part of the problem belongs to PSPACE.
Moreover, it also takes a polynomial amount of space to enumerate all graphs $H$ with $|V(H)|\leq  \lceil \sqrt{\log n} \rceil$ and check if $G$ is $H$-free by brute force. We conclude that {\sc Colour-Path-or-Subgraph} belongs to PSPACE.  

\medskip
\noindent
To prove PSPACE-hardness, we reduce from {\sc $4$-Colour-Path}, which is PSPACE-complete by Theorem~\ref{t-bc}.  Let $(G,\alpha,\beta)$ be an instance of {\sc $4$-Colour-Path}, where $G$ is an $n$-vertex graph and $\alpha$ and $\beta$ are $4$-colourings of $G$.
We now set 
$p=\lceil \sqrt{\log 3n} \rceil$. 
We may assume without loss of generality that $p\geq 5$.
From $(G,\alpha,\beta)$ we construct an instance $(G',\alpha',\beta')$ of {\sc Colour-Path-or-Subgraph}. We first define a graph 
$G^*$ as follows (see also Figure~\ref{f-fig2}): 

\begin{itemize}
\item take $G$;
\item add a clique $K$ of $p-4$ vertices $x_1,\ldots,x_{p-4}$;
\item make each vertex of $K$ adjacent to every vertex of $G$; 
\item add a clique $L$ of four vertices $y_1,\ldots,y_4$;
\item make each vertex of $L$ adjacent to every vertex of $K$ (so $K\cup L$ is a $p$-vertex clique and no vertex of $L$ is adjacent to a vertex of $G$);
\end{itemize}

\begin{figure} [h]
	\resizebox{12cm}{!} {
	\begin{tikzpicture}[main_node/.style={circle,draw,minimum size=1cm,inner sep=3pt]}]
	\node[main_node](g1) at (0,0){};
	\node[main_node](g2) at (1.5,0){$\ldots$};
	\node[main_node](g3) at (3,0){};
	\draw (-1,-1) rectangle (4,2);
	\node[main_node](p1) at (0,3){$x_1$};
	\node[main_node](p2) at (1.5,3){$\dots$};
	\node[main_node](p3) at (3,3){$x_{p-4}$};
	\node[main_node](l1) at (-1,5){$y_1$};
	\node[main_node](l2) at (0.5,5){$y_2$};
	\node[main_node](l3) at (2,5){$y_3$};
	\node[main_node](l4) at (3.5,5){$y_4$};
	\draw(l1)--(l2)--(l3)--(l4);
	\draw (l1) to[out=90, in=90] (l4);
	\draw(l1) to [out=60, in=120](l3);
	\draw(l2) to [out=45, in=135](l4);
	\draw(l1)--(p1);
	\draw(l1)--(p2);
	\draw(l1)--(p3);
	\draw(l2)--(p1);
	\draw(l2)--(p2);
	\draw(l2)--(p3);
	\draw(l3)--(p1);
	\draw(l3)--(p2);
	\draw(l3)--(p3);
	\draw(l4)--(p1);
	\draw(l4)--(p2);
	\draw(l4)--(p3);
	\node(cl) at (-2.5,4.5) {\Large$\alpha^{\prime} \equiv \beta^{\prime}$};
	\node(hc) at (8,1.5) {\Large $\alpha^{\prime} \equiv \beta^{\prime}$};
	\node(ck) at (-2, 3) {\Large $\alpha^{\prime} \equiv \beta^{\prime}$};
	\node(g) at (-2,0) {\Huge$G$};
	\node[circle, draw, minimum size=1.5cm](H1) at (5,0){\Huge$H_1$};
	\node[circle, draw, minimum size=1.5cm](H2) at (7,0){\Huge$H_2$};
	\node[circle, draw, minimum size=1.5cm](H3) at (9,0){$\dots$};
	\node[circle, draw, minimum size=1.5cm](H4) at (11,0){\Huge$H_r$};
	\draw(p1)--(g1);
	\draw(p1)--(g2);
	\draw(p1)--(g3);
	\draw(p2)--(g1);
	\draw(p2)--(g2);
	\draw(p2)--(g3);
	\draw(p3)--(g1);
	\draw(p3)--(g2);
	\draw(p3)--(g3);
	\draw(p1)--(p2)--(p3);
	\draw(p1) to [out=90, in=90] (p3);

	\end{tikzpicture}
}
	\caption{The graph $G^{\prime}$. 
	The graph $G^*$ is the connected component on the left.
	By construction, it holds that $\alpha'\equiv \beta'$ on $V(G')\setminus V(G)$.}\label{f-fig2}

\end{figure}
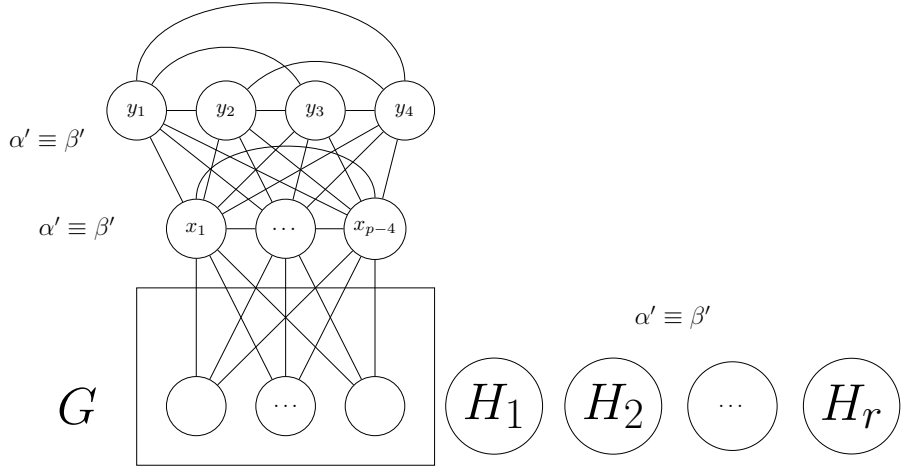

\noindent
Let $\{H_1,\ldots,H_r\}$ be the set of all graphs with exactly $p$ vertices and note that the number of vertices of
$H_1+\ldots + H_r$ is at most $p2^{\frac{p(p-1)}{2}} \leq \lceil \sqrt{\log 3n} \rceil \cdot \sqrt{3n}\leq n$ as $p\geq 5$.  
We now define the graph~$G'$ as the disjoint union of $G^*$ and the graphs $H_1,\ldots,H_r$;
see also Figure~\ref{f-fig2}.
 Note that:
$$|V(G')|=|V(G)|+|K|+|L|+|V(H_1)|+\ldots + |V(H_r)| \leq n+p-4+4+ n < 3n.$$
By adding isolated vertices we may assume that $|V(G')|=3n$.

We now define $\alpha'$ and $\beta'$ as $p$-colourings of $G'$:

\begin{itemize}
\item let $\alpha'=\alpha$ and $\beta'=\beta$ on $G$;
\item for $h\in \{1,\ldots,4\}$, let $\alpha'(y_h)=\beta'(y_h)=h$;
\item for $i\in \{1,\ldots,p-4\}$, let $\alpha'(x_i)=\beta'(x_i)=i+4$;
\item for $j\in \{1,\ldots, r\}$, let $V(H_j)=\{z_1^j,\ldots,z_s^j$\} and let for $q\in \{1,\ldots,s\}$, $\alpha'(z^j_q)=\beta'(z^j_q)=q$.
\end{itemize}

\noindent
We set $\alpha'(u)=\beta'(u)=1$ for each isolated vertex $u$ of $G'$ that we have not yet coloured. 
By construction, $\alpha'$ and $\beta'$ are $p$-colourings of $G'$, in particular because every $H_j$ has $p$ vertices.
We claim that $(G,\alpha,\beta)$ is a yes-instance of {\sc $4$-Colour-Path} if and only if $(G',\alpha',\beta')$ is a yes-instance of {\sc Colour-Path-or-Subgraph}.

First suppose that $(G,\alpha,\beta)$ is a yes-instance of {\sc $4$-Colour-Path}. Then there exists a path from $\alpha$ to $\beta$ in $R_4(G)$. 
We mimic this path in $R_p(G')$, as we can keep the colour $\alpha'(u)=\beta'(u)$ of each vertex $u$ of $G'$ that does not belong to $G$ the same.
Hence, $(G',\alpha',\beta')$ is a yes-instance of {\sc Colour-Path-or-Subgraph}.

Now suppose that $(G',\alpha',\beta')$ is a yes-instance of {\sc Colour-Path-or-Subgraph}. By construction, $G'$ contains every graph on $p$ vertices, and thus every graph on  at most $p$ vertices, as an induced subgraph. 
As $(G',\alpha',\beta')$ is a yes-instance of {\sc Colour-Path-or-Subgraph}
and $p=\lceil \sqrt{\log 3n} \rceil=\lceil \sqrt{\log |V(G')|} \rceil$,
this means that $R_p(G')$ contains
a path $\alpha'\gamma'_1 \cdots \gamma'_t\beta'$ from $\alpha'$ to $\beta'$.
As $\alpha'$ coincides with $\beta'$ on every $H_j$, we may assume without loss of generality that for every $i\in \{1,\ldots,t\}$ and every vertex $z$ of every graph $H_j$, $\gamma'_i(z)=\alpha'(z)=\beta'(z)$. Moreover, for every vertex~$v$ of the clique $K\cup L$, the set of colours used by both $\alpha'$ and $\beta'$ on the vertices of $N(v)\cup \{v\}=\{1,\ldots,p\}$. Hence, these vertices are ``frozen'', that is, we cannot change their colour, so for every $i\in\{1,\ldots,t\}$ and every $v\in K\cup L$ we have that $\gamma'_i(v)=\alpha'(v)=\beta'(v)$. 
Let $\gamma_i$ be the restriction of $\gamma_i'$ to $V(G)$. Then, from the above, we conclude that $\alpha\gamma_1\cdots \gamma_t\beta$ corresponds to a path from $\alpha$ to $\beta$ in $R_4(G)$. Hence, $(G,\alpha,\beta)$ is a yes-instance of {\sc $4$-Colour-Path}.

\medskip
\noindent
We now prove the second part of the theorem. 
Let ${\cal G}$ be a hereditary graph class that is not the class of all graphs.
Then there exists at least one graph $H$ such that every graph $G\in {\cal G}$ is $H$-free. 
Let $\ell=|V(H)|$. We claim that {\sc Colour-Path-or-Subgraph} is constant-time solvable for ${\cal G}$. 
Let $G\in {\cal G}$ be an $n$-vertex graph and let $\alpha$ and $\beta$ be two $p$-colourings of $G$.
If $n \leq 2^{\ell^2}$, then $G$ has constant size and the problem is constant-time solvable.
If $n>2^{\ell^2}$, then $$|V(H)|=\ell < \sqrt{\log n}  \leq \lceil \sqrt{\log n} \rceil.$$ Hence, $(G,\alpha,\beta)$ is a yes-instance of {\sc Colour-Path-or-Subgraph}, as $G$ is $H$-free and $H$ has at most $\lceil \sqrt{\log n} \rceil$ vertices.
\end{proof}

\section{NEXPTIME: Succinct Graph Colouring}\label{s-ex3}

A {\it Boolean circuit} $\phi(x_1,\ldots,x_m,y_1,\ldots,y_m)$ with $2m$ variables defines a graph $G$ on $2^m$ vertices, represented by vectors $(x_1,\ldots,x_m)$ of length~$m$,
according to the rule that there is an edge $(x_1,\ldots,x_m)(y_1,\ldots,y_m)$
between two vertices $(x_1,\ldots,x_m)$ and $(y_1,\ldots,y_m)$
if and only if $\phi(x_1,\ldots,x_m,y_1,\ldots,y_m)$ is true. This allows that some graph families with an exponential number of vertices $2^m$ can be expressed by circuits of size polynomial in $m$. Plainly, this can not be the case in general and indeed graphs whose vertex set is not of size a power of $2$ can only be expressed up to the addition of extra isolated vertices. 

Let us note how any graph on $n$ vertices can be expressed by a circuit with $2n$ variables that has size at most $2n^3$ in what we call the (naive) {\it longhand} method. We will apply this method in the proof of the result in this section.
The $n$ vertices are represented by vectors of length $n$, namely as
$$(1,0,\ldots,0,0),(0,1,\ldots,0,0), \ldots, (0,0,\ldots,0,1),$$ and all of the remaining $2^n-n$ vertices are isolated. If we have an edge $ij$, then this adds a new disjunction to the circuit of the form $$x_i \wedge y_j \wedge \bigwedge_{i \neq \ell \in [n]} \neg x_\ell \wedge \bigwedge_{j \neq \ell \in [n]} \neg y_\ell.$$ Thus, the circuit is in fact in disjunctive normal form.

Suppose we have a graph $G$ represented by a Boolean circuit $\phi(x_1,\ldots,x_m,y_1,\ldots,y_m)$. We can add $k$ vertices to it, again in a longhand way, by expanding the number of variables from $2m$ to $2(m+k)$. In line with our previous longhand method, all vertices other than those of the form  $(x_1,\ldots,x_m,0,\ldots0)$ (the original vertices of $G$) and the $k$ new vertices of the form
\[
\begin{array}{c}
(\overbrace{0,\ldots,0}^{\mbox{$m$ times}},\overbrace{1,0,\ldots,0,0}^{\mbox{$k$ times}})\\
\vdots \\
(\overbrace{0,\ldots,0}^{\mbox{$m$ times}},\overbrace{0,0,\ldots,0,1)}^{\mbox{$k$ times}}\\
\end{array}
\]
are isolated. It is known that the problem {\sc Succinct $3$-Colouring}, which takes as input a Boolean circuit $\phi(x_1,\ldots,x_m,y_1,\ldots,y_m)$ defining a graph $G$ on $2^m$ vertices, and has yes-instances precisely those such that $G$ is properly $3$-colourable, is NEXPTIME-complete. For a proof of this result together with a discussion on succinctly encoded problems we refer to~\cite{Pa94}. We wish to consider the following variant problem.

\problemdef{{\sc Succinct Colouring-or-Subgraph}}{a Boolean circuit $\phi(x_1,\ldots,x_m,y_1,\ldots,y_m)$ defining a graph $G$ on $2^m$ vertices}{is $G$ $\lceil \sqrt{\log m} \rceil$-colourable or $H$-free for some graph $H$ with $|V(H)|\leq  \lceil \sqrt{\log m} \rceil$?} 
Note that, relative to {\sc Colouring-or-Subgraph}, the number of vertices of the graph was mapped to half the number of variables in the circuit.

\begin{theorem}
The {\sc Succinct Colouring-or-Subgraph} problem is $\mathrm{NEXPTIME\mbox{-}complete}$, but constant-time solvable for every hereditary graph class not equal to the class of all graphs.
\end{theorem}

\begin{proof}
We first argue for NEXPTIME membership. Let $G=(V,E)$ be a succinct graph on $2^m$ vertices that is defined by a Boolean circuit $\phi(x_1,\ldots,x_m,y_1,\ldots,y_m)$. Let $p=\lceil \sqrt{\log m} \rceil$.
The question as to whether $G$ is $p$-colourable can be solved in NEXPTIME by guessing the colouring and checking whether adjacent vertices are coloured distinctly. For checking whether $G$ is $H$-free for some graph $H$ with $|V(H)|\leq p$, it suffices to consider only graphs~$H$ on exactly $p$ vertices.
This can be answered, even in EXPTIME, by the following naive algorithm that checks all possibilities of choosing such a graph $H$ one by one. To analyze the running time of this algorithm we observe the following:

\begin{enumerate}
\item the number of graphs on $p$ vertices is at most $2^{\frac{p(p-1)}{2}} \leq m$; and\\[-10pt]
\item checking if a graph $H$ with $p$ vertices is isomorphic to an induced subgraph of $G$ takes $O(p^2|V|^p)
=O( {\log m} \cdot 2^{m \lceil \sqrt{\log m} \rceil}= O(2^{m^2})$ time (we can consider all mappings from $H$ to $G$ by brute force and for each of them we check if edges of $H$ map to edges of $G$).
\end{enumerate}

\noindent
Hence, the total running time of the naive algorithm for checking if $G$ has no graph $H$ on $p$ vertices is
$O(m2^{m^2})$. 

To prove NEXPTIME-hardness we reduce from {\sc Succinct $3$-Colouring}. 
Let $G$ be a succinct graph defined by a Boolean circuit $\phi(x_1,\ldots,x_m,y_1,\ldots,y_m)$. 
We now set
$p=\lceil \sqrt{\log 3m} \rceil$. Add $p-3$ pairwise adjacent vertices to $G$ in the longhand manner discussed above, so the increase in size is at most 
$2(p-3)^3\leq 2\log 3m \sqrt{\log 3m}$.  
We will make the new vertices also adjacent to every vertex of $G$. We can specify the latter by adding to $\phi$ a series of disjuncts for each $j \in \{1,\ldots,p-3\}$, each with $2(p-3)$ variables, encoding the conjunctions $$\bigwedge_{1\leq i \leq p-3} \neg x_{m+i} \wedge \neg \bigwedge_{i\neq j \in \{1,\ldots,p-3\}} \neg y_{m+i} \wedge y_{m+j},$$ of total size at most 
$2(p-3)^2\leq 2\log 3m$.

We now consider the disjoint union $G'$ of the new graph and all possible graphs on $p$ vertices.
Again we do this long hand, so for each graph on $p$ vertices we also add $2^p-p$ isolated vertices to $G'$.
This will require the addition of at most $2p \cdot 2^{\frac{p(p-1)}{2}}\leq 2m$ variables, for sufficiently large $m$, giving a size increase of at most
$16m^3$ as we work in longhand. 

The circuit $\phi'$ specifying  
$G'$ has at most $2m+2(p-3)+2m < 6m$ variables, let us make it up to precisely $6m$, half of which is $3m$. Furthermore, it is of size at most the size of $\phi$ plus $O(m^3)$. By construction, $G'$ contains every graph on $p$ vertices, and thus every graph on at most $p$ vertices, as an induced subgraph. Hence, we deduce in exactly the same way as in the proof of Theorem~\ref{t-0} that
$G'$ is a yes-instance of {\sc Succinct Colouring-or-Subgraph} if and only if $G'$ is $p$-colourable, and that the latter holds if and only if $G$ is $3$-colourable.

\medskip
\noindent
We now prove the second part of the theorem. We do this in the same way as before. 
Let ${\cal G}$ be a hereditary graph class that is not the class of all graphs.
Then there exists at least one graph $H$ such that every graph $G\in {\cal G}$ is $H$-free.  
Let $\ell=|V(H)|$.
We claim that {\sc Succinct Colouring-or-Subgraph} is constant-time solvable for ${\cal G}$. Let $G\in {\cal G}$ be a graph given by a Boolean circuit $\phi(x_1,\ldots,x_m,y_1,\ldots,y_m)$ which has $2^m$ vertices.
If $m \leq 2^{\ell^2}$, then $G$ has constant size and the problem is constant-time solvable.
If $m>2^{\ell^2}$, then $$|V(H)|=\ell < \sqrt{\log m} \leq \lceil \sqrt{\log m} \rceil.$$ Hence $G$ is a yes-instance of {\sc Succinct Colouring-or-Subgraph}, as $G$ is $H$-free and $H$ has at most $\lceil \sqrt{\log m} \rceil$ vertices.
\end{proof}

\end{document}